\documentclass[12pt]{amsart} 

\usepackage[margin = 1in]{geometry}
\usepackage{subfig, graphicx}
\usepackage{amsfonts, amssymb, amsmath}
\begin{document}
\title{A Stock Market Model Based on CAPM and Market Size}
\author{Andrey Sarantsev, Blessing Ofori-Atta, Brandon Flores}
\address{Department of Mathematics and Statistics, University of Nevada, Reno} 
\email{asarantsev@unr.edu} 

\maketitle

\newtheorem{theorem}{Theorem}
\theoremstyle{definition}
\newtheorem{definition}{Definition}
\newtheorem{remark}{Remark}
\newtheorem{example}{Example}

\begin{abstract}
We introduce a new system of stochastic differential equations which models dependence of market beta and unsystematic risk upon size, measured by market capitalization. We fit our model using size deciles data from Kenneth French's data library. This model is somewhat similar to generalized volatility-stabilized models in (Pal, 2011; Pickova, 2013). The novelty of our work is twofold. First, we take into account the difference between price and total returns (in other words, between market size and wealth processes). Second, we work with actual market data. We study the long-term properties of this system of equations, and reproduce observed linearity of the capital distribution curve.  In the Appendix, we analyze size-based real-world index funds. 
\end{abstract}

\thispagestyle{empty}

\section{Introduction}

\subsection{Size effect and the Capital Asset Pricing Model} The size of a stock is measured by its {\it market capitalization}, or {\it market cap:} current stock price multiplied by the number of shares. For a stock portfolio, its market cap is defined as weighted sum of market caps of constituent stocks, with weights equal to the portfolio weights. The size is a very important fundamental characteristic of a stock or a portfolio. 

It is observed that small stocks have higher returns but higher risk than large stocks. An explanation is that small companies are in their dynamic growth phase and they have higher growth potential relative to large mature companies, but small companies are more vulnerable to failure and bankruptcy. Some researchers claim that even after adjusting for risk, small stocks have higher returns than large stocks. This adjustment can be made rigorous within the framework of the Capital Asset Pricing Model (CAPM). Take a portfolio of stocks with total returns (including dividends, not inflation-adjusted) $Q(t)$ during time $t$. Here, we operate in a discrete-time setting. Compare it with risk-free returns from short-term Treasury bills $R(t)$. An investor deserves premium reward for taking risk and investing in stocks rather than in safe Treasury bills. We calculate this {\it equity premium} $P(t)$ by subtracting $P(t) = Q(t) - R(t)$. Next, we compute this equity premium $P_0(t)$ for a market portfolio, used as a benchmark. An example of such benchmark is the Standard \& Poor (S\&P) 500, a widely used benchmark for large U.S. stocks. Run a linear regression:
\begin{equation}
\label{eq:CAPM}
P(t) = \alpha + \beta P_0(t) + \varepsilon(t).
\end{equation}
The parameter $\beta$ shows {\it market exposure}, how much risk the portfolio is exposed to because of fluctuations in the benchmark. The parameter $\alpha$ shows {\it excess return}, how much can one earn from this portfolio on top of this market return. They are often called by their Greek names: {\it beta} and {\it alpha}. The residual $\varepsilon(t)$ is called {\it unsystematic risk} which can be eliminated by diversification. According to the CAPM, $\alpha = 0$ and the only risk which deserves rewards is the systematic risk (due to market exposure) since other risk can be diversified away. The CAPM was proposed in the classic article \cite{Sharpe}. Subsequent research cast doubt on the consistency of CAPM with actual market data. In particular, \cite{Banz} found that taking a portfolio of small stocks generates positive $\alpha$. That is, small stocks have higher returns than large stocks even after adjusting for market exposure (which is greater than 1 for small stocks). Subsequent classic article \cite{3factor} confirmed this. Further research on the size effect can be found in \cite{Semenov} and \cite{SizeSurvey} and references therein. See also critique of CAPM in \cite{CAPM}. 

\subsection{Our model} We study dependence of $\alpha, \beta, \sigma$ (the standard deviation of $\varepsilon(t)$) on the size, measured by the market cap $S(t)$, or, more precisely, by relative size to $S_0(t)$:
\begin{equation}
\label{eq:rel-size}
C(t) = \ln\frac{S_0(t)}{S(t)}.
\end{equation}
We would like to find functions $\alpha, \beta, \sigma$ of $C$ such that for standardized white noise terms $Z(t)$, with $\mathbb E[Z(t)] = 0$ and $\mathbb E[Z^2(t)] = 1$:
\begin{equation}
\label{eq:main}
P(t) = \alpha(C(t)) + \beta(C(t))P_0(t) + \sigma(C(t))Z(t).
\end{equation}
This allows us to quantify how exactly $\alpha$ and $\beta$ (and $\sigma$ the standard deviation of unsystematic risk) depend on the relative size measure. We then consider a version of the equation~\eqref{eq:main} in which equity premia are replaced by {\it price returns}, that is, returns due to price changes (or, equivalently, market cap changes). That is, we replace $P(t)$ and $P_0(t)$ with 
$\ln(S(t+1)/S(t))$ and $\ln(S_0(t+1)/S_0(t))$, ignoring for now risk-free rate and dividends. This gives us:
\begin{equation}
\label{eq:price}
\ln\frac{S(t+1)}{S(t)} = \alpha(C(t)) + \beta(C(t))\ln\frac{S_0(t+1)}{S_0(t)} + \sigma(C(t))Z(t)
\end{equation}
with $C(t)$ from~\eqref{eq:rel-size}. This equation~\eqref{eq:price} includes only market caps of benchmark $S_0(t)$ and the portfolio $S(t)$. This time series equation or its continuous-time version, a stochastic differential equation, allows us to model $S(t)$ and $S_0(t)$ separately from dividends and risk-free returns. On top of this, we add the equation~\eqref{eq:main}. Surprisingly, market data gives us almost the same functions $\alpha, \beta, \sigma$ in~\eqref{eq:main} and~\eqref{eq:price}, but with differing coefficients. Moreover, the white noise terms in~\eqref{eq:main} and~\eqref{eq:price} are almost perfectly (more than 99\%) correlated. This does not follow from any theoretical considerations, and seems simply a piece of good luck which simplifies analysis. 

We also adapt~\eqref{eq:main} and~\eqref{eq:price} for continuous time: Equity premium $P(t)$ becomes then $\mathrm{d}\ln V(t)$, where $V(t)$ is the wealth process adjusted for risk-free returns. More precisely, $V(t) = U(t)/U^*(t)$, where $U(t)$ is the wealth accumulated from investing $U(0) = 1$ in stock portfolio and reinvesting dividends, while $U_*$ is a similar wealth process from investing in Treasury bills. Then~\eqref{eq:main} takes the form
\begin{equation}
\label{eq:premia-CT}
\mathrm{d}\ln V(t) = \alpha(C(t))\,\mathrm{d}t + \beta(C(t))\,\mathrm{d}\ln V_0(t) + \sigma(C(t))\,\mathrm{d}W(t),
\end{equation} 
where $V_0$ is the adjusted wealth process for the benchmark, and $W$ is a Brownian motion: A real-valued continuous process with $W(t) - W(s) \sim \mathcal N(0, t-s)$ independent of $W(u),\, 0 \le u \le s$, for all $0 \le s < t$. This Brownian motion can be viewed as a {\it zoomed out} random walk with very small but very frequent steps. For simplicity, we assume that $\ln V_0(t)$ is also a Brownian motion  with positive drift (which captures the tendency of long-term stock returns to be greater than risk-free returns). Although equity premia have heavy tails and thus are not well-described by the Gaussian distribution, the Brownian motion provides a simple first approximation. Similarly, equation~\eqref{eq:price} becomes
\begin{equation}
\label{eq:prices-CT}
\mathrm{d}\ln S(t) = \alpha(C(t))\,\mathrm{d}t + \beta(C(t))\,\mathrm{d}\ln S_0(t) + \sigma(C(t))\,\mathrm{d}W(t),
\end{equation} 
with $C(t)$ from~\eqref{eq:rel-size}. By the above remark, we can assume that the driving Brownian motion $W$ from~\eqref{eq:premia-CT} and~\eqref{eq:prices-CT} is the same. 

\subsection{Dependence upon the size measure} What is the dependence of market exposure $\beta$, excess return $\alpha$, and standard deviation of unsystematic risk $\sigma$ upon relative size? 

Our statistical analysis does not yield conclusive results. The white noise tests for $Z$ unfortunately fail. Thus we cannot claim that a model~\eqref{eq:main}, ~\eqref{eq:price} passes goodness of fit tests. The most resonable guess for $\alpha, \beta, \sigma$ seems to be:
\begin{align}
\label{eq:data}
\alpha(c), \ \beta(c),\ \sigma(c) = 
\begin{cases}
\alpha_+c,\hspace{1.15cm} 1 + \beta_+c,\ \sigma_+c,\, &c \ge c_+\\
-\alpha_-|c|^{1/2},\ 1 + \beta_-c,\ \sigma_-|c|^{1/2},\, &c \le -c_-.
\end{cases}
\end{align}
Here, $\alpha_{\pm}, \beta_{\pm}, \sigma_{\pm}, c_{\pm}$ are positive constants. Our data analysis does not allow us to suggest these functions in a neighborhood of zero. This is due to the fact that observed $C(t)$ in our data do not come very close to zero.

\subsection{Stochastic Portfolio Theory} The continuous-time version is useful because we can use stochastic calculus and immerse these models in Stochastic Portfolio Theory (SPT). This is a framework for stock market modeling which does not depend on particular models. It suggests overweighing small stocks and continuous rebalancing. This means  investing in small stocks in proportion greater than their market cap would dictate. One example of this is taking the equal-weighted portfolio. The central result of SPT: under mild conditions, diversity (no stock dominates the entire market) and sufficient intrinsic volatility, such portfolios outperform the market portfolio, which invests in each stock in proportion to its market cap; see \cite{BF2008}, \cite{RelArbVS}, \cite{FKK2005}. This theory has solid theoretical basis, and is consistent with observed data. For references, see the book \cite{FernholzBook} and a more recent survey \cite{FernholzKaratzasSurvey}. SPT is based on the same observation as above: small stocks have higher return and risk than large stocks.

Although we mentioned above that SPT is model-independent, there are some SPT models which attempt to capture this observation: {\it competing Brownian particles}, where logarithms of market caps evolve as Brownian motions with drift and diffusion coefficients dependent on their current ranks relative to other particles, \cite{BFK2005}, \cite{5people}, \cite{MyOwn4};  their generalizations with jumps, or with dependence on both name and rank (so-called {\it second-order models}): see articles \cite{5people}, \cite{Levy}, \cite{S2011}; {\it volatility-stabilized models}, where $\ln S(t)$ are modeled by stochastic differential equations (SDE) with volatility inversely proportional to $S(t)$, \cite{Pal}, and their generalizations, \cite{Pickova}. As the number of stocks tends to infinity, the limiting behavior of competing Brownian particles and volatility-stabilized models is studied in \cite{CP2010}, \cite{S2012}, \cite{S2013}.

In this article, we recognize the difference between price and total returns, and model them separately as~\eqref{eq:main} and~\eqref{eq:price} for discrete time or~\eqref{eq:premia-CT} and~\eqref{eq:prices-CT} for continuous time. SPT is much more developed for continuous-time diffusive models based on SDE than for discrete-time, see \cite{PalWong14}. Thus it is reasonable to switch to continuous time in~\eqref{eq:premia-CT} and~\eqref{eq:prices-CT}.

Stochastic Portfolio Theory deals with diversification benefits in the form of excess growth rate, and functionally generated portfolios. The standard assumption in SPT is that there are no dividends; that is, price and total returns are the same. But here this is no longer true. To model separately price returns (which drive the market capitalization processes) and total returns (which drive the wealth processes) requires an extension of the SPT, in particular the concept of functionally generated portfolios. This is left for future research.

\subsection{Data analysis} We take real U.S. market data from Kenneth French's Data Library online: market cap, price and total returns for equal-weighted portfolios made from size deciles (stocks split into top 10\%, next 10\%, etc. according to their size). This library contains data processed from original raw data from the Center for Research in Securities Prices (CRSP) at the University of Chicago, July 1926--June 2020 (84 years), monthly data. 

We study models of $n$ portfolios and the benchmark with market caps $S_0, \ldots, S_n$, relative size measures $C_0, \ldots, C_n$, and wealth processes $V_0, \ldots, V_n$. The corresponding Brownian motions $W_1, \ldots, W_n$ from~\eqref{eq:prices-CT} and~\eqref{eq:premia-CT} (recall that these two equations have the same Brownian motions) are assumed to be i.i.d. for simplicity (although our analysis shows this to be inconsistent with actual data). Consider {\it market weights:} 
$$
\mu_i(t) = \frac{S_i(t)}{S_0(t) + \ldots + S_n(t)},\, i = 0, \ldots, n,
$$
representing the proportion of the $i$th stock in the overall market. Small stocks have smaller market weights. The {\it market weight vector} $\mathbf{\mu} = (\mu_0, \ldots, \mu_n)$ is a Markov process on the $n$-dimensional simplex $\triangle_n$. If this market weight vector converges to a unique stationary distribution as $t \to \infty$ in the total variation norm (see definitions in Section 4), then we call the system {\it stable}. In this article, we state and prove that our market model is stable under certain conditions on $\alpha, \beta, \sigma$. We also investigate whether {\it collisions:} $S_i(t) = S_j(t)$ happen (when small stocks grow and overtake large stocks). Finally, at each time rank weights from top to bottom:
$$
\mu_{(0)}(t) \ge \ldots \ge \mu_{(n)}(t).
$$
The {\it capital distribution curve} is the plot of ranked weights vs their ranks on the double logarithmic scale: 
$$
\left(\ln k, \ln\mu_{(k)}(t)\right),\, k = 0, \ldots, n.
$$
For the actual market data, this plot is linear, except at the endpoint. See \cite[Figure 5.1]{FernholzBook} for the capital distribution curve for the CRSP stock universe for 8 days, chosen to be the last days (December 31) of the eight decades -- December 31, 1929; December 31, 1939; \ldots December 31, 1999. Strikingly, these 8 curves almost completely coincide. In other words, this curve is stable in time. Previously mentioned competing Brownian particles and volatility-stabilized models reproduce this feature, see \cite{CP2010} and \cite{Pal} respectively. In this article, we establish this property for our model using both theoretical analysis and simulations. 

\subsection{Our contributions} We propose a new model within SPT which is consistent with CAPM and captures the size effect. This model is consistent with long-term market data (although not fully consistent) and captures its important features. We study its properties: long-term stability, collisions of particles, market weights, and the capital distribution curve. 

\subsection{Organization of the article} In Section 2, we describe our data in detail and analyze it. This provides a motivation for continuous-time modeling. In Section 3, we introduce a system of SDE as in~\eqref{eq:prices-CT} and~\eqref{eq:premia-CT}, and prove existence and uniqueness of the solution. We also discuss applicability of SPT, since now size and wealth processes are different. In Section 4, we show stability under certain conditions and see how they are compared with data analysis in Section 2.  In Section 5, we replicate the linear capital distribution curve. Finally, Section 6 is devoted to conclusions and suggestions for future research. The Appendix contains some data analysis for existing exchange-traded funds based on size. The code and data can be found on \texttt{GitHub}: \texttt{asarantsev/CAPM-SPT}.

\subsection{Acknowledgements} We thank the Department of Mathematics and Statistics at our University of Nevada, Reno, for welcoming and supportive atmosphere and for fostering research collaboration between faculty and students (undergraduate and graduate). We thank the referees for useful remarks and positive responses. 

\section{Data Analysis}

\subsection{Data description} Our main data source, as mentioned in the Introduction, is the CRSP database at the University of Chicago, taken from Kenneth French's online data library. We take from this library price and total monthly returns for equally-weighted portfolios made of stock in each decile from July 1926 to June 2020. We analyze only the top 8 deciles, corresponding to large-cap, mid-cap, and small-cap stocks. The two bottom deciles are micro-cap stocks which we exclude. Deciles are created based on the market cap at the end of each June. After a year, these deciles are reconstituted. The data library also has average market capitalizations for the end of each month for each decile. Risk-free monthly returns are computed as $R(t) = \ln(1 + r(t)/12)$, where $r(t)$ are monthly data for short-term Treasury bills taken from the Federal Reserve Economic Data (FRED) website: January 1934--June 2020 series TB3MS and July 1926--December 1933 (discontinued) series M1329AUSM193NNBR.

Stock and portfolio returns can be computed in two ways: {\it arithmetic} $A$ and {\it geometric} $G$, which are related as follows: $G = \ln(1 + A)$. Arithmetic returns are quoted regularly for their practicality in computing portfolio returns such as in the data library which we used as our data source. The arithmetic return of a portfolio is equal to the weighted average of arithmetic returns of constituent stocks, but here we convert arithmetic returns to their geometric versions according to the above formula. 

The advantage of geometric over arithmetic returns for our research is apparent when combining compound interest rates. For example, 20\% and 30\% arithmetic returns combined gives 56\% returns, whereas geometric returns combine for 50\% as expected.

Our time unit is equal to a month, of total $T = 12\cdot(2020-1926) = 1128$ months, $t = 0$ corresponding to June 1926, $t = T$ corresponding to June 2020.

\begin{figure}
\centering
\subfloat[$\beta_k(n)-1$ vs $C_k(n)$]{\includegraphics[width = 7cm]{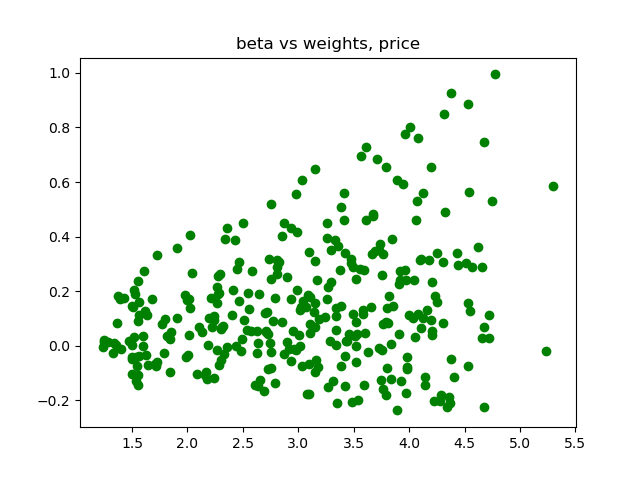}}
\subfloat[$(\beta_k(n)-1)/C_k(n)$ vs $C_k(n)$]{\includegraphics[width = 7cm]{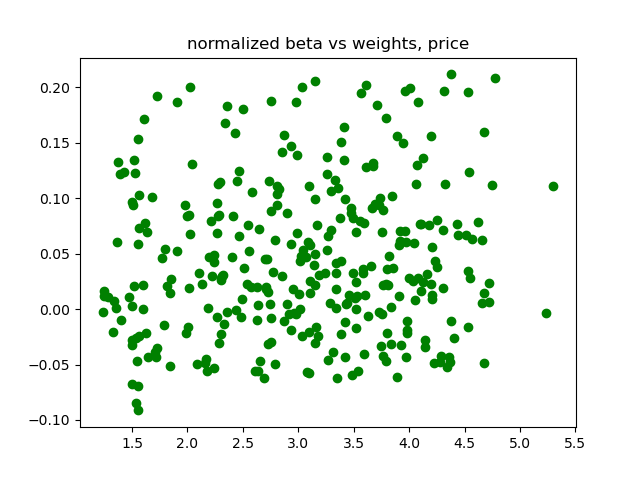}}
\caption{Beta $\beta_k(n)$ vs $C_k(n)$ for top-decile benchmark}
\label{fig:beta-price-top}
\end{figure}

For a decile $k$ (with $k = 1$ for the top decile, $k = 8$ for the bottom decile), its average market cap at end of month $t$ is denoted by $S_k(t)$, the price returns are $Q_k(t)$, and the equity premium (total returns including dividends minus risk-free returns) is $P_k(t)$. 

\subsection{Beta analysis for price returns} Split $T = 1128$ months into $N = 47$ two-year, $K = 24$-month time periods. For each period $n = 1, \ldots, N$ and each decile $k = 1, \ldots, 8$, except the top one, which we use as the benchmark (this top decile roughly corresponds to Standard \& Poor 500 constituent stocks), regress 
\begin{equation}
\label{eq:CAPM-enhanced}
Q_k(t) = \alpha_k(n) + \beta_k(n)Q_1(t) + \delta_k(t),\, k = 2, \ldots, 8,
\end{equation}
where $t$ is in this period, and $\alpha_k(n), \beta_k(n)$ are intercept and slope (excess return and market exposure), found using ordinary least squares; and $\delta_k(t)$ are residuals. Thus we compute beta 
$\beta_k(n)$ for each decile from 2nd to 8th and each of 47 two-year periods. For size measure of the $k$th decile vs top decile, take
$$
C_k(n) = \ln\frac{S_1(Kn)}{S_k(Kn)},
$$
where $S_1(Kn)$ and $S_k(Kn)$ are average market capitalizations for the beginning of $n$th period. Then we plot $\beta_k(n) - 1$ vs $C_k(n)$ and $(\beta_k(n) - 1)/C_k(n)$ in \textsc{Figure}~\ref{fig:beta-price-top}. We plot $\beta_k(n) - 1$ instead of $\beta_k(n)$ since our benchmark for  beta is $1$: The beta for the top decile (which coincides with the benchmark) is $1$. 

\subsection{Model suggestions for price returns} From \textsc{Figure~\ref{fig:beta-price-top} (A)}, we get the suggestion that 
$\beta_k(n) - 1$ is proportional to $C_k(n)$. This is confirmed by \textsc{Figure~\ref{fig:beta-price-top} (B)}. This transformation helps make variance constant. We still cannot claim that these normalized quantities are i.i.d. white noise, since they fail white noise tests. Still, let us extract the trend in $\beta_k(n)$ by replacing it with $1 + \gamma C_k(n)$ for some coefficient. Next, comparing with~\eqref{eq:CAPM-enhanced}, consider the residual $Q_k(t) - (1 + \gamma C_k(n))Q_1(t)$. We make it dependent on $n$ (the overall two-year period), not individual months $t$ in this period, by taking the sum of geometric price returns $Q_k(t)$ over $t$ in this two-year period. This sum is equal to overall price returns $\overline{Q}_k(t)$ in this two-year period. Then we let:
\begin{equation}
\label{eq:CAPM-residual}
\varepsilon_k(n) := \overline{Q}_k(n) - (1 + \gamma C_k(n))\overline{Q}_1(n),\ k = 2, \ldots, 8;\ n = 1, \ldots, N.
\end{equation}

\begin{figure}
\centering
\subfloat[$\varepsilon_k(n)$ vs $C_k(n)$]{\includegraphics[width = 7cm]{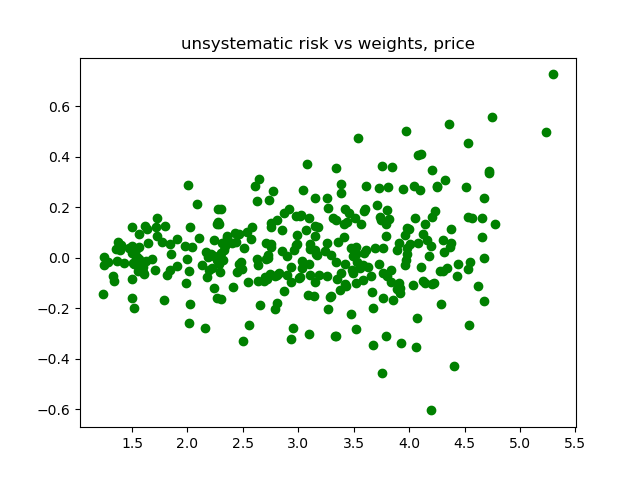}}
\subfloat[$\varepsilon_k(n)/C_k(n)$ vs $C_k(n)$]{\includegraphics[width = 7cm]{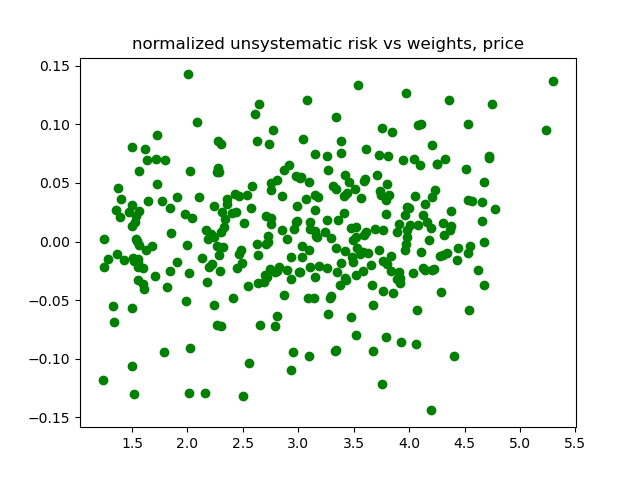}}
\caption{Residual $\varepsilon_k(n)$ vs $C_k(n)$ for top decile  benchmark}
\label{fig:unsys-risk-top}
\end{figure}

We then plot $\varepsilon_k(n)$ vs $C_k(n)$ in \textsc{Figure~\ref{fig:unsys-risk-top} (A)}, together with $\varepsilon_k(n)/C_k(n)$ vs $C_k(n)$ in \textsc{Figure~\ref{fig:unsys-risk-top} (B)}. We see again in \textsc{Figure~\ref{fig:unsys-risk-top} (A)} that $\varepsilon_k(n)$ depends on $C_k(n)$ linearly, and in \textsc{Figure~\ref{fig:unsys-risk-top} (B)} that the variance becomes constant. Again, we cannot claim that $\varepsilon_k(n)/C_k(n)$ are i.i.d. since white noise tests fail. But if we model this as i.i.d. $\mathcal N(\mu, \rho^2)$, we get:
\begin{equation}
\label{eq:residual-final}
\frac{\varepsilon_k(n)}{C_k(n)} = \mu + \rho Z_k(n),
\end{equation}
where $Z_k(n)$ are i.i.d. standard normal random variables. A generalization of this could be:
\begin{equation}
\label{eq:multivariate}
(Z_2(n), \ldots, Z_8(n)) \sim \mathcal N_7(\mathbf{0}, \Sigma)
\end{equation}
i.i.d. multivariate normal with mean zero vector and (not identity) covariance matrix with units on the main diagonal (that is, $\mathbb E[Z_k(n)] = 0$ and $\mathbb E[Z_k^2(n)] = 1$). Combining~\eqref{eq:CAPM-residual} and~\eqref{eq:residual-final}, we get:
\begin{equation}
\label{eq:CAPM-final}
Q_k(n) = (1 + \gamma C_k(n))Q_1(n) + \mu C_k(n) + \rho C_k(n) Z_k(n).
\end{equation}
We can estimate $\gamma = 0.0045, \mu = 0.0069, \rho = 0.052$. This gives us~\eqref{eq:data}, the top row. 

\begin{figure}
\centering
\subfloat[$\beta_k(n) - 1$ vs $C_k(n)$]{\includegraphics[width = 7cm]{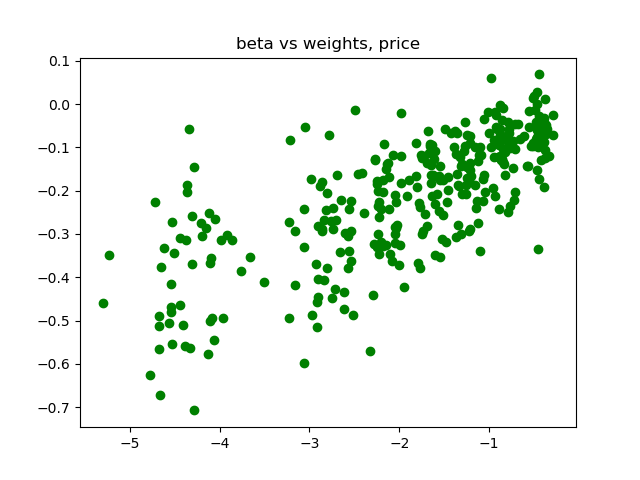}}
\subfloat[$\varepsilon_k(n)/\sqrt{|C_k(n)}$ vs $C_k(n)$]{\includegraphics[width = 7cm]{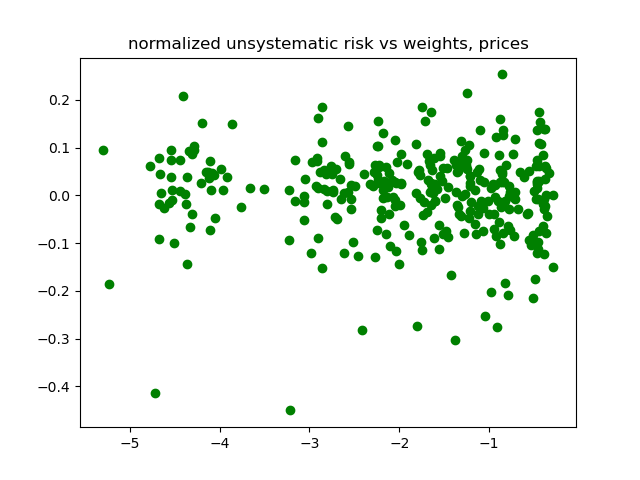}}
\caption{8th decile (bottom) benchmark}
\label{fig:bottom}
\end{figure}

\subsection{Bottom decile as benchmark} If we repeat this analysis in previous subsection with benchmark 8th decile (which in our research is the bottom decile since we ignore the 9th and 10th deciles), we get: For $\beta_k(n)$ in~\eqref{eq:CAPM-enhanced}, we plot $\beta_k(n) - 1$ vs $C_k(n)$ in \textsc{Figure~\ref{fig:bottom} (A)}. In this case, dependence is also linear. Next, take $\gamma$, the mean of these quantities. Create residuals similarly to~\eqref{eq:CAPM-residual}: 
\begin{equation}
\label{eq:CAPM-res-neg}
\varepsilon_k(n) = Q_k(n) - (1 + \gamma\sqrt{|C_k(n)|})Q_8(n),\, k = 1, \ldots, 7,\, n = 1, \ldots, N.
\end{equation} 
The plot of these residuals vs $C_k(n)$, normalized by dividing by $\sqrt{|C_k(n)|}$, is shown in  \textsc{Figure~\ref{fig:bottom} (B)}. Making white noise test for $\varepsilon_k(n)$, we fail to reject this hypothesis. Assume $\varepsilon_k(n) \sim \mathcal N(\mu, \rho^2)$. Combining this with~\eqref{eq:CAPM-res-neg}, we get:
\begin{align}
\label{eq:CAPM-final-neg}
\begin{split}
Q_k(n) &= (1 + \gamma\sqrt{|C_k(n)|})Q_8(n) + \mu\sqrt{|C_k(n)|} + \rho \sqrt{|C_k(n)|} Z_k(n),\\
k &= 1, \ldots, 7,\quad n = 1, \ldots, N;\\
Z_k(n) &= (\varepsilon_k(n) - \mu)/\rho \sim \mathcal N(0, 1)\quad \mbox{i.i.d.}
\end{split}
\end{align}
This verifies~\eqref{eq:data}, the bottom row. From~\eqref{eq:CAPM-final-neg}, we see that linear functions from~\eqref{eq:CAPM-final} are not the only possible and reasonable functions for $\beta$. In Section 4, we state and prove results for general $\alpha, \beta, \rho$. Our coefficients are $\gamma = 0.12, \mu = 0.0055, \rho = 0.090$. 

\subsection{Data analysis for equity premia} Repeating similar analysis for equity premia instead of price returns, we get similar functions for both cases: top and bottom deciles as benchmarks. Moreover, point estimates for price returns and equity premia are close; $\gamma = 0.045, \mu = 0.0017, \rho= 0.052$ for the top decile benchmark, and $\gamma = 0.12, \mu = 0.0024, \rho = 0.088$ for the 8th decile benchmark. Thus functions $\beta$ and $\sigma$ for price returns and equity premia are the same. Unfortunately, for $\alpha$ this is not true. We will have separate functions for price returns and equity premia. As noted in Section 1, white noise terms $Z_k(n)$ for price returns and equity premia are almost perfectly correlated, since the Pearson correlation coefficient is greater than 99\%. However, price returns and equity premia for the benchmark are not perfectly correlated. Thus noise terms for size and wealth processes are different.

\subsection{Conclusions} Data analysis in previous subsections implies that excess return for small stocks is positive, and market exposure for small stocks is greater than $1$. Large stocks have negative excess return and market exposure less than $1$. In other words, small stocks are riskier than larger stocks, but they have higher returns even after adjusting for market exposure. We will accommodate this for continuous time in the next section. 

\section{A Continuous-Time Model}

\subsection{Model description} Take a filtered probability space $(\Omega, \mathcal F, (\mathcal F_t)_{t \ge 0}, \mathbb P)$ with the filtration satisfying the {\it usual conditions} (each $\mathcal F_t$ contains all $\mathbb P$-null sets; that is, $A \in \mathcal F_t$ with $\mathbb P(A) = 0$ and $B \subseteq A$ implies $B \in \mathcal F_t$; and the filtration is right-continuous, that is, $\mathcal F_t = \cap_{s > t}\mathcal F_s$ for all $t \ge 0$). All stochastic processes $X = (X(t),\, t \ge 0)$ below are {\it adapted}, that is, $X(t)$ is $\mathcal F_t$-measurable for every $t \ge 0$. We say that a stochastic process $B = (B(t),\, t \ge 0)$ is a {\it standard Brownian motion} if $B(t) - B(s) \sim \mathcal N(0, t-s)$ is independent of $\mathcal F_s$ for every $0 \le s < t$.

Replace white noise terms $Z_k(n)$ in~\eqref{eq:CAPM-final} or~\eqref{eq:CAPM-final-neg} by standard Brownian motions $W_k(t)$, or, more exactly, its differential $\mathrm{d}W_k(t)$. Replace price returns $Q_k(n)$ by $\mathrm{d}\ln S_k(t)$. Recall that $S_k(t)$ is the average market capitalization of the $k$th portfolio at time $t$. Replace equity premia $P_k(n)$ with $\mathrm{d}\ln V_k(t)$, where $V_k(t)$ is the wealth (including reinvested dividends) of the $k$th decile, starting from $V_k(0) = 1$, divided by wealth similarly computed from risk-free returns. Set three functions: $\alpha, \beta, \sigma$ of $C_k(t)$ (alpha, beta, and standard deviation of unsystematic risk). Set a separate function $\alpha_*$ of $C_k(t)$ for equity premia instead of price returns. We shall not set separate functions for $\beta$ and $\sigma$ since from our data analysis we found that these functions are the same for price returns and equity premia. We have $n$ portfolios indexed by $1, \ldots, n$, and the benchmark indexed with $0$ (that is, total $n+1$ portfolios). We write separate functions for price returns and equity premia (which we denote by asterisks). Consider the following system of SDE:
\begin{equation}
\label{eq:SDE-prices}
\mathrm{d}\ln S_k(t) = \alpha(C_k(t))\,\mathrm{d}t + \beta(C_k(t))\,\mathrm{d}\ln S_0(t) + \sigma(C_k(t))\,\mathrm{d}W_k(t),
\end{equation}
\begin{equation}
\label{eq:SDE-premia}
\mathrm{d}\ln V_k(t) = \alpha_*(C_k(t))\,\mathrm{d}t + \beta(C_k(t))\,\mathrm{d}\ln V_0(t) + \sigma(C_k(t))\,\mathrm{d}W_k(t),
\end{equation}
where $C_k(t)$ is the {{\it relative size measure} defined in the Introduction:
\begin{equation}
\label{eq:size}
C_k(t) = \ln\frac{S_0(t)}{S_k(t)}.
\end{equation}
We also model $(S_0, V_0)$ using two-dimensional geometric Brownian motion: $(\ln S_0, \ln V_0)$ has drift vector $(g_S, g_V)$ and covariance matrix
$$
\mathbf{\Sigma}_{(S, V)} = 
\begin{bmatrix}
\sigma_S^2 & \rho_0\sigma_S\sigma_V\\
\rho_0\sigma_S\sigma_V & \sigma_V^2
\end{bmatrix}
$$
That is, for $t > s$ we have: $\ln S_0(t) - \ln S_0(s) \sim \mathcal N(g_S(t-s), \sigma_S^2(t-s))$ and $\ln V_0(t) - \ln V_0(s) \sim \mathcal N(g_V(t-s), \sigma^2_V(t-s))$; here, $\rho_0$ is the correlation between these two increments. Thus
\begin{align}
\label{eq:2D}
\begin{split}
\ln S_0(t) &= \ln S_0(0) + g_St + \sigma_SW_S(t),\\
\ln V_0(t) &= \ln V_0(0) + g_Vt + \sigma_VW_V(t),
\end{split}
\end{align}
where $W_S$ and $W_V$ are (correlated) standard Brownian motions. We assume $W_1, \ldots, W_n$ are independent of $(B_S, B_V)$. Note that $W_1, \ldots, W_n$ can be dependent of each other standard Brownian motions. We assume that $W = (W_1, \ldots, W_n)$ is an $n$-dimensional Brownian motion with zero drift vector and covariance matrix $\mathbf{\Sigma}_W$ with units on the main diagonal. 

\begin{definition} For functions $\alpha, \beta, \sigma : \mathbb R \to \mathbb R$, real numbers $g_S, g_V$, $2\times 2$ covariance matrix $\mathbf{\Sigma}_{(S, V)}$, and $n\times n$ correlation matrix $\mathbf{\Sigma}_W$, the system of  equations~\eqref{eq:SDE-prices},~\eqref{eq:SDE-premia},~\eqref{eq:size},~\eqref{eq:2D} is called a {\it CAPM-size market model} of $N+1$ portfolios, indexed by $0, \ldots, N$. The process $S_k$ is called {\it market size}, or {\it market cap}, of the $k$th portfolio; and the process $V_k$ is called the {\it wealth process} for this $k$th portfolio.  The portfolio indexed by $k = 0$ is called the {\it benchmark}. The functions $\alpha$ and $\beta$ are called by their Greek names. The function $\sigma$ is called the {\it standard error} (of the diversifiable risk). 
\end{definition}

\subsection{Existence and uniqueness}  We can rewrite~\eqref{eq:SDE-prices} using~\eqref{eq:size}:
\begin{equation}
\label{eq:size-SDE}
\mathrm{d}C_k(t) = -\alpha(C_k(t))\,\mathrm{d}t + (1 - \beta(C_k(t))\,\mathrm{d}\ln S_0(t) + \sigma(C_k(t))\,\mathrm{d}W_k(t).
\end{equation}

\begin{remark} Each $k$th equation~\eqref{eq:SDE-prices} is independent of other equations: Does not contain $S_l$ for other $l  = 1, \ldots, n$. It depends only on $S_k$ and $S_0$.
\end{remark}

\begin{definition} Define the {\it explosion time} as follows:
\begin{equation}
\label{eq:explosion}
\mathcal T := \mathcal T_1\wedge\ldots\wedge\mathcal T_n,
\end{equation}
where $\mathcal T_k$ is the explosion time for~\eqref{eq:SDE-prices}, or, equivalently, for~\eqref{eq:size-SDE}. 
\end{definition}

On $[0, \mathcal T)$, equations~\eqref{eq:SDE-prices} and~\eqref{eq:SDE-premia} have a unique solution. 

\begin{theorem} (a) If $\alpha, \beta, \sigma, \alpha_* : \mathbb R \to \mathbb R$ are measurable and locally bounded, there exists a unique strong solution to the system of SDE~\eqref{eq:SDE-prices},~\eqref{eq:SDE-premia},~\eqref{eq:size}, until the explosion time $\mathcal T$.

\smallskip

(b) If, in addition, we have the following linear bound:
$$
|\alpha(c)| + |\beta(c)| + |\sigma(c)| + |\alpha_*(c)| \le K(1 + |c|),
$$
then the explosion time is infinite: $\mathcal T = \infty$. 
\label{thm:first}
\end{theorem}

\begin{proof} (a) First, let us show existence and uniqueness in~\eqref{eq:size-SDE}. The diffusion part in this equation is given by
$\tilde{\sigma}(C_k(t))\,\mathrm{d}\tilde{W}_k(t)$, where 
$$
\tilde{\sigma}^2(c) := \sigma^2_S(1-\beta(c))^2 + \sigma^2(c),
$$
and $\tilde{W}_k$ is a standard Brownian motion which is a combination of $W_S$ and $W_k$. The drift part is given by $\tilde{\gamma}(C_k(t))\,\mathrm{d}t$, where 
\begin{equation}
\label{eq:drift}
\tilde{\gamma}(c) := -\alpha(c) + g_S(1 - \beta(c)).
\end{equation}
Thus we can rewrite~\eqref{eq:size-SDE} as follows:
\begin{equation}
\label{eq:new-SDE}
\mathrm{d}C_k(t) = \tilde{\gamma}(C_k(t))\,\mathrm{d}t + \tilde{\sigma}(C_k(t))\,\mathrm{d}\tilde{W}_k(t).
\end{equation}
These functions $\tilde{\gamma}$ and $\tilde{\sigma}$ are measurable, and we apply \cite{Zvonkin} to prove strong existence and pathwise uniqueness of the solution $C_k$ to this SDE, at least until the explosion time $\mathcal T_k$. Here we use Remark 1: the system of $n$ equations $C_1, \ldots, C_n$ consists of $n$ independent one-dimensional SDE. Next, we can reconstruct $S_1, \ldots, S_n$ from $C_1, \ldots, C_n$ and $S_0$: $S_k(t) = S_0(t)e^{-C_k(t)}$. Since $S_0$ is well-defined for infinite time horizon (as geometric Brownian motion), the strong solution $S_k$ for~\eqref{eq:SDE-prices} exists and is pathwise unique, too, at least until the explosion time $\mathcal T_k$. Finally, strong existence and pathwise uniqueness for~\eqref{eq:SDE-premia} can be shown as follows. Let 
$$
\mathcal T_{k, m} := \inf\{t \ge 0\mid |C_k(t)| = m\},\, k = 1, \ldots, n;\ m = 1, 2, \ldots
$$
Since $\alpha_*, \beta, \sigma$ are locally bounded, they are bounded on $[-m, m]$. We can rewrite~\eqref{eq:SDE-premia} as
\begin{equation}
\label{eq:V}
\mathrm{d}\ln V_k(t) = \overline{\gamma}(C_k(t))\,\mathrm{d}t + \overline{\sigma}(C_k(t))\,\mathrm{d}\overline{W}_k(t),
\end{equation}
with $\overline{W}_k$ a standard Brownian motion, and 
\begin{equation}
\label{eq:two-f}
\overline{\gamma}(c) := \alpha_*(c) + g_V\beta(c),\quad \overline{\sigma}^2(c) := \sigma^2_V\beta^2(c) + \sigma^2(c).
\end{equation}
These two functions from~\eqref{eq:two-f} are bounded on $[-m, m]$. Rewriting~\eqref{eq:V} as
$$
\ln V_k(t) = \int_0^t \overline{\gamma}(C_k(s))\,\mathrm{d}s + \int_0^t\overline{\sigma}(C_k(s))\,\mathrm{d}\overline{W}_k(s),
$$
we see that (strong) solution until $t < \mathcal T_{k, m}$ is well-defined and pathwise unique, because the integrals (both usual and stochastic) are finite. As $m \to \infty$, we have: $\mathcal T_{k, m} \uparrow \mathcal T_k$ almost surely. Thus~\eqref{eq:SDE-premia} have a pathwise unique strong solution until $\mathcal T_k$. Combining this with~\eqref{eq:explosion}, we complete the proof of (a).

\smallskip

Next, (b) follows from~\eqref{eq:size-SDE} and~\eqref{eq:V} and estimates of linear growth for $\tilde{\gamma}$ and $\tilde{\sigma}$:
$$
|\tilde{\gamma}(c)| + |\tilde{\sigma}(c)| \le K_*(1 + |c|),
$$
following from similar estimates for $\alpha, \alpha_*, \beta, \sigma$. 
\end{proof}

\section{Main Results}

\subsection{Long-term stability results} Recall the definition of market weights.

\begin{definition} For the market model in~\eqref{eq:SDE-prices}, we define {\it market weights} as follows:
$$
\mu_i(t) = \frac{S_i(t)}{S_0(t) + \ldots + S_n(t)},\, i = 0, \ldots, n,\, t \ge 0.
$$
\end{definition}

These market weights sum up to $1$ for every $t \ge 0$, and are in one-to-one correspondence with $C_1(t), \ldots, C_n(t)$: There exists a bijection $\Phi : \triangle_n \to \mathbb R^n$ such that 
$$
\Phi : (m_0, \ldots, m_n) \mapsto \left(\ln\frac{m_0}{m_1}, \ldots, \ln\frac{m_0}{m_n}\right).
$$
The process $\mathbf{C} = (C_1, \ldots, C_n)$, as its each component, is a Markov process. The same is true for the market weight vector $\mu = (\mu_0, \ldots, \mu_n)$. 

\begin{definition} We call a probability measure $\pi_{\mu}$ on $\triangle_n$ a {\it stationary distribution} if $\mu(0) \sim \pi_{\mu}$ implies $\mu(t) \sim \pi_{\mu}$ for all $t \ge 0$. The market is called {\it stable} if this market weight vector has a unique stationary distribution $\pi_{\mu}$, and when we start from another initial distribution $\mu(0)$, then the distribution of $\mu(t)$ converges to $\pi_{\mu}$ as $t \to \infty$ in the {\it total variation} (TV) norm:
$$
\sup\limits_{D \subseteq \triangle_n}|\mathbb P(\mu(t) \in D\mid \mu(0) = x) - \pi_{\mu}(D)| \to 0\quad \mbox{as}\quad t \to \infty.
$$
\end{definition}

We can similarly (and equivalently) define stability for the process $\mathbf{C} = (C_1, \ldots, C_n)$, which is $\mathbb R^n$-valued. These definitions of stability are equivalent, since there is a one-to-one continuous mapping between $C$ and $\mu$. 

\begin{theorem}
Under assumptions of Theorem~\ref{thm:first}, suppose
\begin{equation}
\label{eq:limits}
\varliminf\limits_{c \to \infty}\left[\alpha(c) + g_S(\beta(c) - 1)\right] > 0\quad \mbox{and}\quad \varlimsup\limits_{c \to -\infty}\left[\alpha(c) + g_S(\beta(c) - 1)\right] < 0.
\end{equation}
Then the market system is stable. 
\label{thm:stable}
\end{theorem}

\begin{proof} We apply the method of Lyapunov functions. As a Lyapunov function, take the following infinitely differentiable function $V : \mathbb R \to [0, \infty)$:
$$
V(c) := \begin{cases} |c|,\, c \ge 2;\\ 0,\, c \le 1.\end{cases}
$$
This function can be constructed by smoothing kernel and convolution. The generator for each $C_k$ in~\eqref{eq:SDE-prices} is given by
$$
\mathcal Lf(c) := \tilde{\gamma}(c)f'(c) + \frac12\tilde{\sigma}^2(c)f''(c).
$$
Recall~\eqref{eq:drift}: $\varlimsup\limits_{c \to \infty}\tilde{\gamma}(c) < 0$ and $\varliminf\limits_{c \to -\infty}\tilde{\gamma}(c) > 0$. Combining this with~\eqref{eq:limits}, we get: 
$$
\varlimsup\limits_{|c| \to \infty}\mathcal LV(c) < 0.
$$
From the articles \cite{MT1993a} and \cite{MT1993b}, we get tightness of each $C_k$. That is, $\sup_{t \ge 0}\mathbb P(|C_k(t)| \ge c) \to 0$ as $c \to \infty$. The same is true for the vector $\mathbf{C} = (C_1, \ldots, C_n)$. Combining this observation with the following property: for each $i$, 
$\mathbb P(a < C_i(t) < b\mid C_i(0) = x) > 0,\quad x, a, b \in \mathbb R$, we complete the proof. \end{proof}

\begin{remark} We also have convergence 
\begin{align*}
\sup\limits_{g \in  \mathcal G}&|\mathbb E[g(C_k(t))] - (\pi_C, g)| \to 0,\, t \to \infty,\\
\mathcal G &:= \{g : \mathbb R \to \mathbb R\mid \sup\limits_{z \in \mathbb R}\frac{|g(z)|}{1 + |z|} < \infty\}.
\end{align*}
In particular, we have $\mathbb E[C_k(t)] \to m_C$, where $m_C$ is the mean of the distribution $\pi_C$.
\end{remark}

\begin{example} If $\beta(c) = 1 + \gamma c$, $\alpha(c) = \mu c$, then $\alpha(c) + g_S(\beta(c) - 1) = (\mu + g_S\gamma)c$, and~\eqref{eq:limits} is equivalent to 
\begin{equation}
\label{eq:cond}
\Gamma := \mu + g_S\gamma > 0.
\end{equation}
\label{exm:linear}
\end{example}

\begin{example} From the data analysis in Section 2, let
$$
\alpha(c) := 
\begin{cases}
\alpha_+|c|^{\gamma_+},\ c \ge 0;\\
-\alpha_-|c|^{\gamma_-},\, c \le 0;
\end{cases}
\ 
\beta(c) := 1 + 
\begin{cases}
\beta_+c,\, c \ge 0;\\
\beta_-c,\, c \le 0;
\end{cases}
\ 
\sigma(c) := 
\begin{cases}
\sigma_+|c|^{\gamma_+},\, c \ge 0;\\
\sigma_-|c|^{\gamma_-},\, c \le 0.
\end{cases}
$$
Here, $\alpha_{\pm}, \beta_{\pm}, \gamma_{\pm}, \sigma_{\pm} > 0$. Let us find conditions for~\eqref{eq:limits}: For the first condition, if $\gamma_+ < 1$, we have $g_S\beta_+ > 0$; if $\gamma_+ = 1$, we have $\alpha_+ + g_S\beta_+ > 0$; if $\gamma_+ > 1$, we have $\alpha_+ > 0$. Similarly for the second condition in~\eqref{eq:limits}. The actual estimates from Section 2 satisfy these conditions. 
\label{exm:power}
\end{example}

\subsection{Hitting times} We wish to allow $S_i$ and $S_0$ to exchange ranks. That is, we want to allow  for $S_i(0) < S_0(0)$ but $S_i(t) > S_0(t)$ for some $t > 0$, or vice versa. This is consistent with real world market behavior, when portfolios exchange ranks based on size. In terms of relative size measures $C_i$, we wish that $C_i$ can move from positive half-line to negative half-line. In particular, it must hit zero with positive probability. 

This is not true for Example~\ref{exm:linear} if $\sigma(c) = \rho c$ for $\rho > 0$. Indeed, then $C_i$ is a geometric Brownian motion with drift $-\Gamma$ and thus converges to $0$ almost surely as $t \to \infty$. Thus the limiting stationary distribution is the delta measure at the origin,  $\delta_{(0, \ldots, 0)}$. The corresponding limiting distribution for $\mu$, the market weight vector is $\delta_{(1/(n+1), \ldots, 1/(n+1))}$. If $\sigma(c)$ is bounded away from zero, this changes the behavior of $\mathbf{C}$.

\begin{theorem} Under conditions of Theorem~\ref{thm:stable}, assume 
\begin{equation}
\label{eq:sigma}
\sigma(c) \ge \sigma_* > 0\quad \mbox{for all}\quad c \in \mathbb R.
\end{equation}

(a) For every $c \in \mathbb R$ with positive probability there exists a $t > 0$ such that $C_i(t) = c$. 

\smallskip

(b) The stationary distribution $\pi_C$ for $\mathbf{C} = (C_1, \ldots, C_n)$ has support on $\mathbb R^n$. The stationary distribution $\pi_{\mu}$ for 
$\mu$ has support on $\Delta_n$. 
\label{thm:hit}
\end{theorem}

\begin{proof} Return again to the equation~\eqref{eq:new-SDE}, which is a simplified equation~\eqref{eq:size-SDE}. 

\smallskip

(a) Compute the scale function $s$ for the diffusion $C_k$: Its derivative is 
$$
s'(c) = \exp\left[-2\int_0^c\frac{\tilde{\gamma}(u)}{\tilde{\sigma}^2(u)}\,\mathrm{d}u\right].
$$
There exist $\gamma_* > 0$ and $c_* > 0$ such that 
$$
\tilde{\gamma}(u) \le -\gamma_*,\, c \ge c_*;\quad \tilde{\gamma}(u) \ge \gamma_*,\, c \le -c_*.
$$
Moreover, $\tilde{\sigma}(u) \ge \sigma(u) \ge \sigma_*$ for all $u \in \mathbb R$. Thus for $c \ge c_*$, 
\begin{equation}
\label{eq:scale-plus}
s'(c) \ge s'(c_*)\exp\left[2\int_{c_*}^c\frac{\gamma_*}{\sigma_*^2}\,\mathrm{d}c\right] = s'(c_*)\exp\left[2(c-c_*)\frac{\gamma_*}{\sigma_*^2}\right].
\end{equation}
A similar estimate is true for $c \le -c_*$:
\begin{equation}
\label{eq:scale-minus}
s'(c) \le s'(-c_*)\exp\left[2(|c|-c_*)\frac{\gamma_*}{\sigma_*^2}\right].
\end{equation}
The speed measure has bounded Lebesgue density $1/(s'(c)\sigma^2(c))$, as shown in~\eqref{eq:scale-plus},~\eqref{eq:scale-minus},~\eqref{eq:sigma}. Apply Feller's test from \cite[Chapter 5, Section 5]{KSBook} to complete the proof. 

\smallskip

(b) The statement for $\pi_C$ follows from ellipticity of the elliptic partial differential equation governing the Lebesgue density of this stationary distribution. The statement for $\pi_{\mu}$ follows from one-to-one mapping between $\mathbf{C}$ and $\mu$. 
\end{proof}

It seems to us that~\eqref{eq:sigma} is a reasonable assumption, since we would want to allow for exchange of ranks of portfolios. This does not contradict our statistical analysis, since we can observe only $C_i(t) > c_+$ or $C_i(t) < -c_-$. Our suggested functions $\sigma(c)$ do give $\sigma(0) = 0$ if we extend them to $[-c_-, c_+]$ as is. But we could not observe $C_i(t)$ in a neighborhood of zero, thus we can extend it as a piecewise function. 

Under assumptions of Theorem~\ref{thm:hit}, the stationary distribution for each $C_i$ has (after normalization) density as above: $1/(s'\sigma^2)$, which is supported on the whole real line but is bounded. This is different from the case $\sigma(c) = \sigma_0c$ discussed above in Example~\ref{exm:linear}, when the stationary distribution is concentrated at one point, but the components of the stationary distribution for the overall vector $C$ are not independent since the SDE for individual relative size measures are dependent. 

\subsection{Lack of propagation of chaos} For interacting particle systems, sometimes dependence (as a process itself, or stationary distribution) vanishes as the number of particles tends to infinity. This phenomenon is called {\it propagation of chaos}, since the system becomes less interdependent and more chaotic. Such results were shown for competing Brownian particles and volatility-stabilized models (see citations in the Introduction). But it is unreasonable to expect this for the current system since all particles are dependent. The limiting density, if it exists, will likely be a solution to a stochastic partial differential equation. To derive this large system limit is left for future research.

\section{Capital Distribution Curve}

\subsection{Modified plots} Let us study the capital distribution curve $(\ln k, \ln \mu_{(k)}(t))$ in this model. We solve the system of stochastic differential equations explicitly. We use the values of $C_k(t)$ to plot the capital distribution curve: 
$$
C_k(t) = \ln\frac{S_0(t)}{S_k(t)} = \ln\frac{S_0(t)}{S(t)} - \ln \mu_k(t).
$$
Thus the ranking of $C_k(t)$ reverts the ranking of ranked market weights:
$$
0 \le C_{(1)}(t) \le \ldots \le C_{(n)}(t).
$$
Thus we can plot the modified curve $(\ln k, C_{(k)}(t))$. If this curve is linear, the same can be said for the original capital distribution curve. Now we shall study the system of stochastic differential equations~\eqref{eq:size-SDE} and plot $(\ln k, C_{(k)}(t))$ for fixed $t$. 

\subsection{Degenerate case} Even if the system is stable, under conditions of Theorem 2, capital distribution curve can be degenerate, equal to one point: $(1/(n+1), \ldots, 1/(n+1))$. This is the case when $C_k(t) \to 0$ a.s. as $t \to \infty$ for each $k = 1, \ldots, n$. Indeed, in this case $\mu_k(t) \to 1/(n+1)$ a.s. as $t \to \infty$.  In particular, this is true in Example 1 with $\sigma(c) = \rho c$ for $\rho > 0$. Below, we consider the case when the capital distribution curve is not trivial. 

\subsection{Linear case} This is the case when
\begin{equation}
\label{eq:linear-case}
\alpha(c) = \mu c,\quad \beta(c) = 1 + \gamma c,\quad \sigma(c) = \rho.
\end{equation}
Then the system~\eqref{eq:size-SDE} is linear:
$$
\mathrm{d}C_k(t) = -\mu C_k(t)\,\mathrm{d}t - \gamma C_k(t)(g_S\,\mathrm{d}t + \sigma_S\,\mathrm{d}W_S(t)) + \rho\,\mathrm{d}W_k(t).
$$
As shown in \cite[Chapter 5, Section 6.C]{KSBook}, we can solve this system explicitly:
\begin{align*}
C_k(t) &= Z(t)\left[C_k(0) - \int_0^tZ^{-1}(u)(\rho\,\mathrm{d}W_k(u) - \gamma\sigma_S\rho\,\mathrm{d}u)\right];\\
Z(t) &= \exp\left[-(\mu + \gamma g_S)t - \gamma\sigma_S\,W_S(t) - \frac12\gamma^2\sigma_S^2t\right].
\end{align*}
Below, we show the simulation results for $n = 100$, with initial conditions $C_k(0) = 0$, $k = 1, \ldots, n$. We take estimated values $\mu = 0.0069$, $\gamma = 0.0045$. For $\rho$, we take the value $0.1$, which is consistent with estimates.  Finally, estimates for mean $g_S$ and standard deviation $\sigma_S$ of monthly price returns for the benchmark (top decile) are given by $g_S = 0.0044$ and $\sigma_S = 0.0541$. We simulate until $t = 100$, assuming all $W_1, \ldots, W_n, W_S$ are independent Brownian motions. 

\begin{figure}
\centering
\includegraphics[width = 10cm]{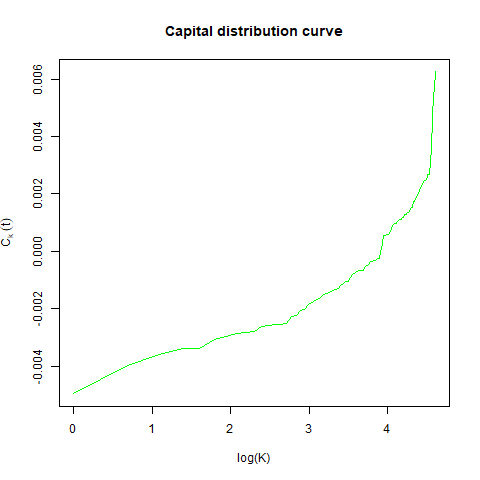}
\caption{Capital Distribution Curve $(\ln k, C_{(k)})$}
\end{figure}

\section{Conclusions} We developed a model in this article which can be viewed as an enhancement of the Capital Asset Pricing Model, which stresses dependence of stock portfolios upon the overall market. The portfolios are based on size (market cap), and the quantities $\alpha$, $\beta$, and standard error $\sigma$ (of diversifiable idiosyncratic risk) depends on size (more precisely, size of portfolio relative to the size of benchmark). Thus we wrote a system of stochastic differential equations. 

We write separately systems of equations for equity premia (total returns, including dividends, minus risk-free returns), and for market size (that is, price returns). They are very similar and the idiosyncratic risk can be taken the same. The equity premium and price returns of the benchmark are driven by different random processes (although correlated). 

Using CRSP 1926--2020 monthly data of size deciles, we find reasonable guesses for $\alpha$, $\beta$, $\sigma$ as functions of relative size. Our statistical analysis is not fully rigorous, because it fails white noise tests. However, we do find some reasonable results. 

On the theoretical side, we prove long-term stability results: under some conditions on $\alpha, \beta, \sigma$, the vector of relative size measures converges as $t \to \infty$ to a stationary distribution. 

Finally, an important feature of real-world markets: stability and linearity of the capital distribution curve, is reproduced in our model by numerical simulation. 

Future research can include making more sophisticated time series models which take into account autocorrelations, or non-Gaussian fluctuations of the market. It seems important to develop the SPT for the case of dividends, when price and total returns (and by extension market capitalizations and wealth processes) are different. Finally, we would like to derive a large system limit.

\section*{Appendix. Statistical Analysis of Size-Based Index Funds} 

These size deciles of the CRSP universe are not directly investable. But there exist size-based funds available for individual investors. Among many of them, let us take \texttt{JKJ}, \texttt{JKG}, \texttt{JKD}: iShares Morningstar Small-Cap, Mid-Cap, and Large-Cap exchange-traded funds. These are based on largest 70\%, next 20\%, and next 7\% of the total universe of stocks. In other words, Large-Cap corresponds to Deciles 1--7 weighted by their market capitalizations, Mid-Cap corresponds to Deciles 8--9 weighted by their market capitalizations, and Small-Cap corresponds to the top 7\% of the bottom Decile 10.

Monthly total arithmetic returns for these funds are taken from BlackRock web site, July 2004 -- August 2020. For risk-free returns, we take 1-month Treasury Constant Maturity Rate from Federal Reserve Economic Data web site, observed at the last day of each month June 2004 -- July 2020. We compute geometric versions of these returns. From each such rate $r$, we obtain geometric total monthly returns for the next month $\ln(1 + r/1200)$. Then we compute equity premia $P_S, P_M, P_L$ for these funds. Regress the first two upon the third:
\begin{align}
\label{eq:joint}
\begin{cases}
P_S(t) = \alpha_S + \beta_SP_L(t) + \varepsilon_S(t);\\
P_M(t) = \alpha_M + \beta_MP_L(t) + \varepsilon_M(t);
\end{cases}
\quad 
\begin{bmatrix}
\varepsilon_S(t)\\
\varepsilon_M(t)
\end{bmatrix}
 \sim \mathcal N_2\left(\begin{bmatrix}0 \\  0\end{bmatrix}, \Sigma\right).
\end{align}
The quantile-quantile plots, Shapiro-Wilk and Jarque-Bera normality tests, and autocorrelation function plots allow us to assume that each series of residuals can be modeled by i.i.d. normal distribution. Thus we can apply standard Student tests for regression coefficients. The 95\% confidence intervals for each of $\alpha_S$ and $\alpha_M$ contain zero. Thus we can assume that $\alpha_S = \alpha_M = 0$, but the confidence intervals for $\beta_S$ and $\beta_M$ do not contain $1$. Point estimates of these coefficients are: $\beta_S = 1.27$, $\beta_M = 1.15$. Estimates for standard errors for residuals, and cross-correlation between residuals are: $\sigma_S = 0.026$, $\sigma_M = 0.019$, $\rho = 0.83$. The $R^2$ values for each regression are 87\% and 80\%. Thus we see that the CAPM works for actual traded size-based funds.

\end{document}